\documentclass{amsart}
\usepackage[T1]{fontenc}
\usepackage{amsfonts}
\usepackage{amssymb}
\usepackage{multirow} 
\usepackage{longtable} 
\usepackage{array} 
\usepackage{makecell} 
\usepackage{threeparttable}
\usepackage{tablefootnote}
\usepackage[colorlinks=true,linkcolor=blue,allcolors=blue
]{hyperref}
\usepackage{amsmath,blkarray,booktabs,bigstrut}
\usepackage{tikz}
\usepackage{tikz-cd}
\usepackage{graphicx}
\usepackage{dutchcal}
\usepackage{algorithm}
\usepackage{algpseudocode}
\usepackage{soul}
\usepackage{cleveref}
\numberwithin{equation}{section}

\theoremstyle{plain}
  \newtheorem{theorem}[subsection]{Theorem}
  \newtheorem{proposition}[subsection]{Proposition}
  \newtheorem{lemma}[subsection]{Lemma}
  \newtheorem{corollary}[subsection]{Corollary}

\theoremstyle{definition}

\newcommand\bR{{\mathbb R}}
\newcommand{\ip}[2]{\left<#1,#2\right>}
\newcommand{\norm}[1]{\left\Vert#1\right\Vert}
\newcommand{\abs}[1]{\left\vert#1\right\vert}
\newcommand{\rar}{\rightarrow}
  
\begin{document}
\title[Taylor Diagram and Wasserstein Distance]{Taylor Diagram and Wasserstein Distance for Model Evaluation}

\author{Dongwei Chen}
\address{Department of Mathematics and Statistical Science, University of Idaho, Moscow, ID, USA, 83844}
\email{dongweic@uidaho.edu}

\author{Emily J. King}
\address{Department of Mathematics, Colorado State University, Fort Collins, CO, USA, 80523}
\email{emily.king@colostate.edu}

\author{Hungjui Yu}
\address{Cooperative Institute for Research in the Atmosphere, Colorado State University, Fort Collins, CO, USA, 80523; Department of Atmospheric Science, Colorado State University, Fort Collins, CO, USA, 80523}
\email{hungjui@colostate.edu}

\begin{abstract}
 The Taylor diagram is used to evaluate and compare predictive models with observed data and has many applications in climate and environmental sciences. 
 Three statistics of interest--the centered root-mean-squared error, standard deviation, and the product-moment correlation coefficient--are related by the law of cosines; Taylor diagrams leverage this fact to allow visualization of all three statistics in a two-dimensional plot without any loss of information.
 In this work, we present a novel model evaluation tool--the Wasserstein-Taylor diagram--formed from integrating the Wasserstein distance from optimal transport into a Taylor diagram framework.
 This tool is built upon the fact that the law of cosines still holds in the Taylor diagram if the centered root-mean-squared error and  product-moment correlation coefficient are replaced by the centered 2-Wasserstein distance and quantile correlation coefficient, respectively. 
 The advantage of this Wasserstein-Taylor diagram is that one can use the distribution perspective afforded by the Wasserstein distance to compare observed and predicted data sets of different sizes and even to compare statistical models with only observed data. 
 We further show that the quantile correlation coefficient on empirical measures converges to the quantile correlation coefficient on the sampled measures.
\end{abstract}


\maketitle

\section{Introduction and Main Results}

The Taylor diagram is an effective visualization tool for summarizing how well a model prediction matches an observation or ground truth using three key statistics: product-moment correlation coefficient, standard deviation, and centered root-mean-square error (centered RMSE). These metrics are displayed together in a single polar plot, enabling multiple models or experiments to be compared against the same reference at a glance \cite{taylor-2001}. 
In climate and atmospheric sciences, the Taylor diagram has become a standard approach for multi-model evaluation and for assessing the impact of different configurations or parameterizations on weather and climate model performance \cite{gleckler-2008, jiang-2012, khosravi-2025}. 
The Taylor diagram is able to represent these three evaluation metrics in a two-dimensional plot due to the law of cosines. 
The azimuthal angle corresponds to the inverse cosine of the product-moment correlation coefficient between the model and observations, the radial distance from the origin indicates the standard deviation of the model prediction, and the Euclidean distance between a model point and the reference point is the centered RMSE. 
This construction allows diverse model results to be visualized simultaneously, facilitating direct comparison of pattern similarity, amplitude of variability, and overall discrepancy in an intuitive and concise format \cite{taylor-2001}. 
The mathematical construction of the Taylor diagram is given as follows.

Suppose $X=(x_1,\dots,x_N)$ is a non-constant sequence of ground truth observations and  $Y=(y_1,\dots,y_N)$ is a corresponding non-constant sequence of model predictions. 
Define the means as 
\[
\overline{x} := \frac1N\sum\limits_{i=1}^N x_i \quad \textrm{and} \quad\overline{y}: = \frac1N\sum\limits_{i=1}^N y_i.
\]
Then the associated centered data sequences are given by $X'=(x'_1,\dots,x'_N)$ and $Y'=(y'_1,\dots,y'_N)$ where
$$
x_i' = x_i - \overline{x} \ \text{and} \ y_i' = y_i - \overline{y},
$$
yielding the associated sample variances 
$$
\sigma_x^2 = \frac1N \sum_{i=1}^N (x_i')^2>0 \ \text{and} \
\sigma_y^2 = \frac1N \sum_{i=1}^N (y_i')^2>0.
$$
It is well-known that the  product-moment correlation coefficient is just the cosine similarity between the centered data $X'$ and $Y'$; that is, 
\begin{equation}\label{eqn:defncorrcoeff1}
\rho := \rho(X,Y) = \frac{\frac1N\sum_{i=1}^N x_i'y_i'}{\sigma_x \sigma_y} = \frac{\langle X', Y' \rangle_{\ell^2(N)}}{ \|X'\|_{\ell^2(N)} \|Y'\|_{\ell^2(N)}} = \cos (\theta),
\end{equation}
where $\theta$ is the angle between centered data $X'$ and $Y'$. 
Also note that the centered RMSE between $X$ and $Y$ is given by
\[
E := E(X',Y') = \left( \frac1N \sum_{i=1}^N (x_i' - y_i')^2 \right)^{1/2}=\frac{1}{\sqrt{N}} \norm{X'-Y'}_{\ell^2(N)}.
\]
By expanding the formula defining the centered RMSE, we have 
\begin{equation}\label{TaylorDiagram}
E^2 
 = \sigma_x^2 + \sigma_y^2 - 2\,\frac1N\sum_{i=1}^N x_i'y_i'= \sigma_x^2 + \sigma_y^2 - 2 \sigma_x \sigma_y \rho 
 = \sigma_x^2 + \sigma_y^2 - 2\sigma_x\sigma_y\cos(\theta),
\end{equation}
which indeed satisfies the law of cosines in a triangle with side lengths $E$, $\sigma_x$, and $\sigma_y$. 
The relationship between $E$, $\sigma_x$, $\sigma_y$, and $\rho$ is visualized in Figure~\ref{fig:lawcos}.
\begin{figure}[ht]
    \centering
    \includegraphics[width=0.5\linewidth]{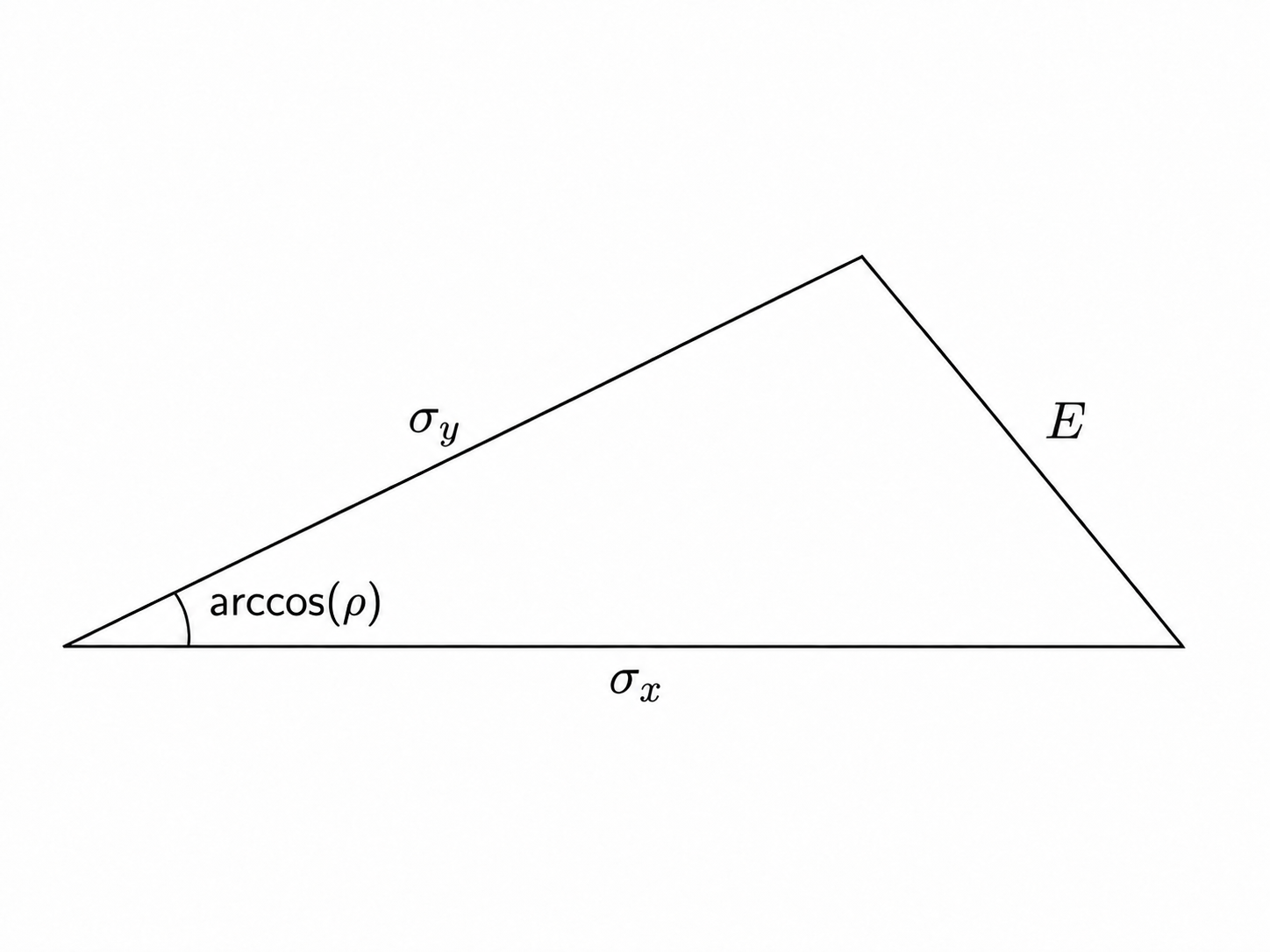}
    \caption{Centered RMSE $E$, standard deviation of observations $\sigma_x$, standard deviation of model predictions $\sigma_y$, and the inverse cosine of the product-moment correlation coefficient $\arccos(\rho)$ satisfy the law of cosines.}
    \label{fig:lawcos}
\end{figure}

Figure \ref{fig:Taylor-example} shows an example of a Taylor diagram. 
For each model, $\sigma_x$ is the same as it is the standard deviation of the ground truth.
So, we can view the Taylor diagram as resulting from gluing together triangles which all share a base and which are of the form seen in Figure~\ref{fig:lawcos}.
\begin{figure}[ht]
    \centering
    \includegraphics[scale=0.5]{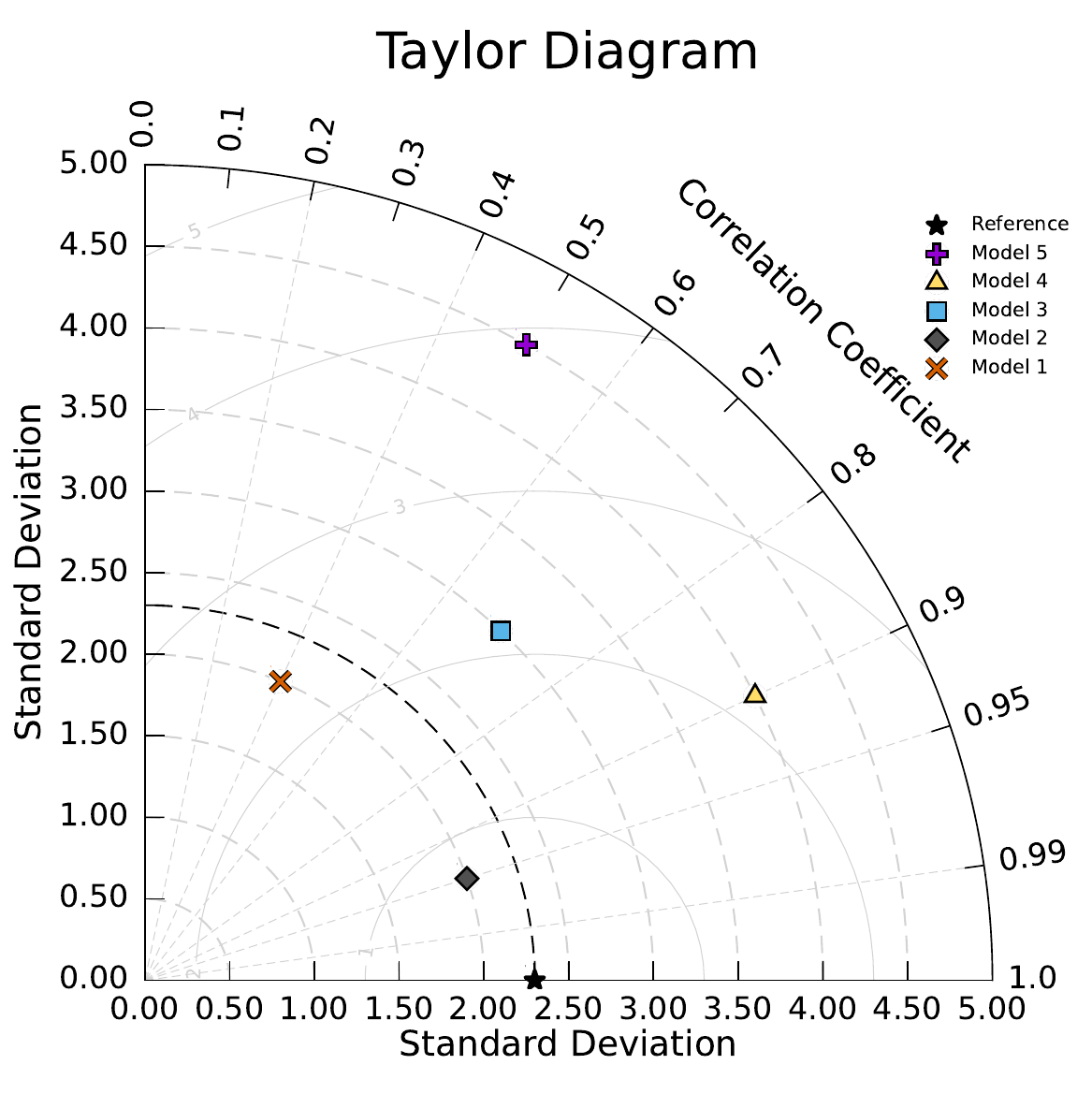} 
 \caption{An example of a Taylor diagram.  In the diagram, the standard deviation (STD) of the reference data or ground truth is 2.3. The STD of the output data of Model 1 is 2, and the associated  product-moment correlation coefficient with the reference data is 0.4. For Model 2, the STD is still 2 but the product-moment correlation coefficient is 0.95. The STD of Model 3 is 3 and the corresponding product-moment correlation coefficient is 0.7.  Similarly,  the STD of Model 4 is 4 and the product-moment correlation coefficient is 0.9. Finally, Model 5 has the largest STD as 4.5 while the product-moment correlation coefficient is 0.5. Since Model 2 has the closest STD to the reference STD, highest product-moment correlation coefficient, and is closest to the reference point, Model 2 has the best performance compared to other models according to the Taylor-diagram criterion. 
}    
 \label{fig:Taylor-example}
\end{figure}
The Taylor diagram allows an easy visual comparison of models, including those which may share a statistical value, like having the same centered RMSE.

Note that the Taylor diagram requires that both observed and predicted data sequences have the \emph{same} number of elements (and that those elements are paired in a particular order). 
Otherwise, the product-moment correlation coefficient and centered RMSE are not well-defined. 
However, due to factors that arise often in practice, such as sensor errors, there are typically many missing values in observed data. 
This means that the observed and predicted data sequences have \emph{different} sizes, and thus the Taylor diagram cannot be used in this scenario. 
It is natural to ask how to compare the performance of different models even when some observations are missing and the observed and model-predicted data sequences have different sizes.   
In particular, in some cases, such as wind speed probability distribution modeling \cite{yang2025beta}, we only have the observed data and are interested in finding an appropriate statistical distribution to model the data, which naturally requires the comparison of different distributions. 
The above analysis indicates that the Taylor diagram does not work in many cases, including those that have different sample sizes,
different temporal resolutions, and missing data cases, as well as in the case in which one wants to compare model distributions with the empirical distribution.

The main contribution of this work is the introduction of Wasserstein-Taylor diagrams, which lift the framework of Taylor diagrams to apply to probability distributions using a Wasserstein metric.
The geometric representation relies on Lemma~\ref{WTaylorDiagram} (a special case of a decomposition in \cite{irpino2007optimal}); namely, that the law of cosines still holds if the centered RMSE is replaced by the centered $2$-Wasserstein distance and the product-moment correlation coefficient with the quantile correlation coefficient. 
The former is a metric between centered measures with finite second moments and the latter the cosine similarity between quantile functions of those centered measures.
The advantage is that one can use the Wasserstein distance to compare observed and predicted data sequences with different sizes or when order does not matter and to compare statistical distributions when only observed data are available. 
We further prove in Theorem~\ref{thm:samplestat} that the quantile correlation coefficients of empirical measures converge to the quantile correlation coefficient of the sampled measure almost surely.

This paper is organized as follows. 
We begin by introducing the necessary background concerning the Wasserstein distance and optimal transport in Section \ref{sec:preliminary}. 
Next, in Section \ref{sec:mainTheorem}, we introduce Wasserstein-Taylor diagrams and explain their motivation and advantages.
The fact that the law of cosines holds for any two probability measures with finite second moments and positive variance allows us to create this two-dimensional visualization in Wasserstein-Taylor diagrams.
We thus give a self-contained proof that the law of cosines holds for Wasserstein distance and quantile correlation in Lemma \ref{WTaylorDiagram}. 
Then we show in Proposition~\ref{prop:sort} that the quantile correlation coefficient of measures assigning $1/N$ times a Dirac delta to points in data sequences of cardinality $N$ maximizes the product-moment correlation of all possible permutations of the support.
Next in Theorem~\ref{thm:samplestat}, we prove that the quantile correlation coefficient of empirical measures converges almost surely to the quantile correlation coefficient of the sampled measures.

We present applications of the Wasserstein-Taylor diagram in Section \ref{sec:application}. 
We first evaluate machine learning models for convective cloud fraction in climate science in Section~\ref{sec:climate}. 
In this use, two interpretable methods are more clearly separated under the distributional criterion of the Wasserstein-Taylor diagram than the Taylor diagram.
Then in Section~\ref{sec:wind} we use the Wasserstein-Taylor diagram to evaluate the wind speed distribution fitting. 
We also introduce the so-called normalized Wasserstein-Taylor diagram in Section~\ref{sec:normalized} to compare the relative performance of distributions for several variables in a single diagram. 
Finally, in Section \ref{sec:summary}, we provide a summary and discussion.

\section{Preliminaries }\label{sec:preliminary}
Let $\mathcal{P}(\mathbb{R})$ be the set of Borel probability measures on the Euclidean space $\mathbb{R}$ and let $\mathcal{P}_2(\mathbb{R}) \subset \mathcal{P}(\mathbb{R})$  be the set of Borel probability measures with finite $2$nd moment, i.e., 
\begin{equation*}
   \int_{\mathbb{R}} \vert x \vert ^2 d\mu(x) < + \infty.
\end{equation*}

If $\mu \in  \mathcal{P}(\mathbb{R}) $ and $f: \mathbb{R} \rightarrow \mathbb{R}$ is a Borel measurable map, then $f_{\#} \mu \in \mathcal{P}(\mathbb{R})$ is called the \textit{pushforward} of $\mu$ by  the map $f$ and is defined as
$$f_{\#} \mu (A) : = \mu \big (f^{-1}(A) \big ) \ \text{for any Borel set}  \ A \subset \mathbb{R}.$$ 
Given $\mu \in  \mathcal{P}(\mathbb{R})$, we define the  \emph{centered measure} of $\mu$ as
\begin{equation*}
   \widetilde{\mu} := {\tau_\mu}_\# \mu,
\end{equation*}
where $\tau_\mu(x) = x -m_\mu$ and $m_\mu$ is the mean of measure $\mu$ given by
$$m_\mu = \int_{\mathbb{R}} x d\mu(x).$$

Furthermore, given two probability measures $\mu, \nu \in \mathcal{P}(\mathbb{R})$, $\Gamma(\mu,\nu)$ is denoted as the set of transport couplings with marginals $\mu$ and $\nu$, that is,
$$ \Gamma(\mu,\nu) :=  \left \{ \gamma \in \mathcal{P}(\mathbb{R} \times \mathbb{R}): {\pi_{{ x}}}_{\#} \gamma = \mu, \ {\pi_{{ y}}}_{\#} \gamma = \nu \right \},$$
where $\pi_{{ x}}$, $\pi_{{ y}}$ are orthogonal projections onto the $x$ and $y$ coordinates: for any $(x,y) \in \mathbb{R} \times \mathbb{R} $, $\pi_{{ x}}(x, y) = x$ and $ \pi_{{ y}}(x,y) = y$. 

Then, for $\mu,  \nu \in \mathcal{P}_2(\mathbb{R})$, the $2$-Wasserstein distance $ W_2(\mu,\nu)$ is used to quantify the similarities and discrepancies between $\mu$ and $\nu$:  
$$ W_2(\mu,\nu) := \left(\underset{\gamma \in \Gamma(\mu,\nu)}{\operatorname{inf}}  \int_{\mathbb{R} \times \mathbb{R}}  \left \vert x-y\right \vert^2 \ d\gamma(x,y)\right)^{1/2}$$
and is thus of interest in optimal transport.
It is well-known that the $2$-Wasserstein distance is a metric on $\mathcal{P}_2(\mathbb{R})$. 
See~\cite{figalli2021invitation} for details on optimal transport and Wasserstein distances.

Let $\mu\in \mathcal{P}_2(\mathbb{R})$. 
Then its \emph{quantile function} $F^{-1}: (0,1) \rar \bR$ is defined as
\begin{equation*}
    F^{-1} (t) = \inf \left\{x\in\mathbb{R}: \mu(-\infty, x]\geq t\right\}.
\end{equation*}
The quantile function  is the generalized inverse of the cumulative distribution function, which explains the $-1$ exponent in the notation.  
We note a few key properties of quantile functions 
 which we will make use of. Recall that if $U$ is uniformly distributed on $(0,1)$, then the inverse 
transform sampling theorem gives $F^{-1}(U)\sim\mu$. Consequently, for every 
$\mu$-integrable function $g$, we have 
$$\int_{\mathbb{R}} g(x) d\mu(x)=\int_0^1 g(F^{-1}(t)) dt.$$ 
In particular, for $\mu\in\mathcal{P}_2(\mathbb{R})$ with quantile function $F^{-1}$ and its centered measure $\tilde{\mu}$ with quantile function $\tilde{F}^{-1}$,
\[
m_\mu = \int_0^1 F^{-1}(t)\, dt \quad \textrm{and} \quad  \sigma_\mu^2  = \norm{F^{-1}}_{L^2(0,1)}^2 - m_\mu^2 = \norm{\tilde{F}^{-1}}_{L^2(0,1)}^2
\]
for $m_\mu$ the expected value of $\mu$ and $\sigma_\mu$ the standard deviation.
Additionally, if $\nu\in \mathcal{P}_2(\mathbb{R})$ has quantile function $G^{-1}$, then
\[
W_2^2(\mu,\nu)
= \int_0^1 |F^{-1}(t) - G^{-1}(t)|^2 dt.
\]

Quantile functions may also be used to define a correlation between measures in $\mathcal{P}_2(\bR)$ called the \emph{quantile correlation coefficient} \cite{irpino2015basic,irpino2015linear}.
The quantile correlation coefficient $\rho_q$ between $\mu$ and $\nu$ in $\mathcal{P}_2(\bR)$ with $\sigma_\mu, \sigma_\nu >0$ is just the cosine similarity between the quantile functions $\widetilde{F}^{-1}$ and $\widetilde{G}^{-1}$ for the centered measures $\widetilde{\mu}$ and $\widetilde{\nu}$: 
\begin{equation}\label{eqn:defncorrcoeff2}
\rho_q = \rho_q(\mu, \nu) = \cos (\theta) = \frac{\langle \widetilde{F}^{-1}, \widetilde{G}^{-1}\rangle_{L^2(0,1)}}{\| \widetilde{F}^{-1}\|_{L^2(0,1)} \| \widetilde{G}^{-1}\|_{L^2(0,1)}} = \frac{\langle \widetilde{F}^{-1}, \widetilde{G}^{-1}\rangle_{L^2(0,1)}}{\sigma_\mu \sigma_\nu},
\end{equation}
where $\theta$ is the angle between $\widetilde{F}^{-1}$ and $\widetilde{G}^{-1}$ in $L^2(0,1)$, and $\sigma_\mu$ and $\sigma_\nu$ are standard deviations of $\mu$ and $\nu$, respectively. 
Note that given $\mu, \nu \in \mathcal{P}_2(\mathbb{R})$ with quantile functions $ F^{-1}$ and $G^{-1}$, Irpino and Romano in \cite{irpino2007optimal}  proposed the following decomposition for $W_2^2(\mu,\nu)$:
\begin{align}
    W_2^2(\mu, \nu)  &= \underbrace{(m_\mu - m_\nu)^2}_{\textrm{location}} + \underbrace{(\sigma_\mu - \sigma_\nu)^2}_{\textrm{size}}+ \underbrace{2 \sigma_\mu \sigma_\nu (1-\rho_q)}_{\textrm{shape}} \nonumber\\ 
    &= (m_\mu - m_\nu)^2 + W_2^2(\tilde{\mu},\tilde{\nu}).\label{eq:W2Decomposition}
\end{align}
Also see \cite{irpino2015linear, schefzik2021fast} for more applications of such a decomposition. 
The quantile correlation coefficient appears in the fixed marginal literature as the maximum correlation of a joint bivariate distribution with marginals $\mu$ and $\nu$ (e.g., \cite{lin2014recent}). The general upper bound goes back to Hoeffding~\cite{hoeffding1940masstabinvariante} and Fr{\'e}chet \cite{frechet1951correlation}.
Note that the term ``quantile correlation'' sometimes refers to a different statistic \cite{choi2022quantile}.

\section{Wasserstein-Taylor Diagram}\label{sec:mainTheorem}

The Taylor diagram admits
a law-of-cosines decomposition involving the standard deviation, the  product-moment correlation coefficient, and centered RMSE.  
In this section, we show that an analogous identity holds for the 
centered $2$-Wasserstein distance between two probability measures
using the closed-form quantile representation of the Wasserstein distance. This yields the so-called \emph{Wasserstein-Taylor diagram}, where the standard deviations, centered $2$-Wasserstein distance, and the quantile correlation coefficient satisfy the law of cosines, as illustrated in Figure \ref{fig:comparison-explaination}.

\begin{figure}[ht]
    \centering
    \includegraphics[scale=0.8]{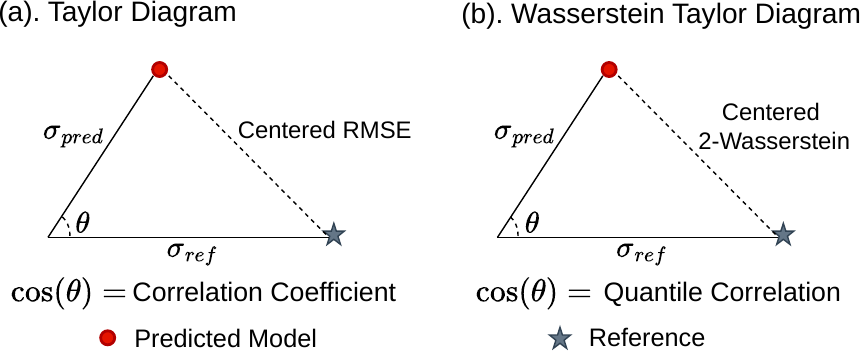} 
    \caption{Comparison of Taylor diagrams and Wasserstein-Taylor diagrams. $\sigma_{pred}$ and  $\sigma_{ref}$ are the standard deviations of the predicted model and the reference ground truth, respectively. In the Taylor diagram,  $\sigma_{pred}$,  $\sigma_{ref}$, and centered root-mean-square error (centered RMSE) satisfy the law of cosines, and $\cos(\theta)$ is the product-moment correlation coefficient between predicted data and reference ground truth data. In the Wasserstein-Taylor diagram,  $\sigma_{pred}$,  $\sigma_{ref}$, and centered 2-Wasserstein distance (between the predicted data distribution and the reference data distribution) satisfy the law of cosines, and $\cos(\theta)$ is the quantile correlation.}
    \label{fig:comparison-explaination}
\end{figure}

Similar to the decomposition of the Taylor diagram, Lemma \ref{WTaylorDiagram} shows the geometric and mathematical framework of the law of cosines in the Wasserstein-Taylor diagram.
Although Lemma \ref{WTaylorDiagram} follows from the decomposition in Equation~\eqref{eq:W2Decomposition} (which was proposed in \cite{irpino2007optimal}) applied to centered measures (i.e., those with $m_\mu = m_\nu = 0$), we include a proof here as the geometric interpretation vis-\`a-vis the law of cosines is critical to the construction of the Wasserstein-Taylor diagram.

\begin{lemma}\label{WTaylorDiagram}
Suppose $\mu, \nu \in \mathcal{P}_2(\mathbb{R})$ with $\sigma_\mu,\sigma_\nu>0$ and quantile functions $ F^{-1}$ and $G^{-1}$ and centered measures $\widetilde{\mu} $ and $ \widetilde{\nu}$. Then we have the following law of cosines 
\begin{equation}\label{eqn:WTaylorDiagram}
    W_2^2(\widetilde{\mu}, \widetilde{\nu})  = \sigma_\mu^2 + \sigma_\nu^2 -2 \sigma_\mu \sigma_\nu \cos(\theta) = \sigma_\mu^2 + \sigma_\nu^2 -2 \sigma_\mu \sigma_\nu \rho_q ,
\end{equation}
where $\theta$ is the angle between quantile functions $\widetilde{F}^{-1}$ and $\widetilde{G}^{-1}$ for $\widetilde{\mu}$ and $ \widetilde{\nu}$  in $L^2(0,1)$, and $\rho_q$ is the quantile correlation coefficient between $\mu$ and $\nu$. 
\end{lemma}
\begin{proof}
Note that the mean values of $\mu$ and $\nu$ are given as
$$
m_\mu = \int_0^1 F^{-1}(t)  dt \ \text{and} \ m_\nu = \int_0^1 G^{-1}(t) dt.
$$
Then the associated centered measures are given as $   \widetilde{\mu} = {\tau_\mu}_\# \mu $  and $ \widetilde{\nu} =  {\tau_\nu}_\# \nu$
where $\tau_\mu(x) = x -m_\mu$ and $\tau_\nu(x) = x -m_\nu$. So the quantile functions for $\widetilde{\mu}$ and $ \widetilde{\nu}$ are
$$
\widetilde{F}^{-1}= F^{-1} - m_\mu \ \text{and} \  \widetilde{G}^{-1} = G^{-1} - m_\nu.
$$
And the corresponding standard deviations for $\mu$ and $\nu$ are given by
$$
\sigma_\mu^2 = \int_0^1 |\widetilde{F}^{-1}(t)|^2 dt \ \text{and} \ \sigma_\nu^2 = \int_0^1 |\widetilde{G}^{-1}(t)|^2 dt.
$$ 
The closed form for the $2$-Wasserstein distance using quantile functions shows that
\[
\begin{aligned}
    W_2^2(\widetilde{\mu}, \widetilde{\nu}) = \int_0^1 |\widetilde{F}^{-1}(t) - \widetilde{G}^{-1}(t)|^2 dt &=  \int_0^1 |\widetilde{F}^{-1}(t)|^2 + |\widetilde{G}^{-1}(t)|^2 -2\widetilde{F}^{-1}(t)\widetilde{G}^{-1}(t)dt \\
    & = \sigma_\mu^2 + \sigma_\nu^2 -2 \int_0^1 \widetilde{F}^{-1}(t)\widetilde{G}^{-1}(t)dt \\
    & = \sigma_\mu^2 + \sigma_\nu^2 -2 \langle \widetilde{F}^{-1}, \widetilde{G}^{-1}\rangle_{L^2(0,1)}. 
\end{aligned}
\]
Now define  
$$
\cos (\theta) := \frac{\langle \widetilde{F}^{-1}, \widetilde{G}^{-1}\rangle_{L^2(0,1)}}{\| \widetilde{F}^{-1}\|_{L^2(0,1)} \| \widetilde{G}^{-1}\|_{L^2(0,1)}} = \frac{\langle \widetilde{F}^{-1}, \widetilde{G}^{-1}\rangle_{L^2(0,1)}}{\sigma_\mu \sigma_\nu},
$$
where $\theta$ is the angle between $\widetilde{F}^{-1}$ and $\widetilde{G}^{-1}$ in $L^2(0,1)$. 
Then 
\begin{equation*}
    W_2^2(\widetilde{\mu}, \widetilde{\nu})  = \sigma_\mu^2 + \sigma_\nu^2 -2 \sigma_\mu \sigma_\nu \cos (\theta).
\end{equation*}
Therefore, we have the law of cosines. 
Note that the cosine similarity between quantiles $\widetilde{F}^{-1}$ and $\widetilde{G}^{-1}$ is just the quantile correlation coefficient $\rho_q$:
$$
\rho_q = \cos (\theta) = \frac{\langle \widetilde{F}^{-1}, \widetilde{G}^{-1}\rangle_{L^2(0,1)} }{\sigma_\mu \sigma_\nu}.
$$
Then
\begin{equation*}
    W_2^2(\widetilde{\mu}, \widetilde{\nu})  = \sigma_\mu^2 + \sigma_\nu^2 -2 \sigma_\mu \sigma_\nu \rho_q.
\end{equation*}
\end{proof}
We point out that just as the product-moment correlation requires data to be centered, the quantile correlation is computed on centered distributions.
In particular, the location term $(m_\mu - m_\nu)^2$ in Equation~\eqref{eq:W2Decomposition} plays no role in the diagram.
We present an idea in the end of the discussion in Section~\ref{sec:summary} on integrating this information in a Wasserstein-Taylor diagram if it is pertinent for your application objectives.

Now that we have established the law of cosines in Lemma~\ref{WTaylorDiagram}, we are able to generalize the use of law of cosines in Taylor diagrams to allow evaluation of models from probability distributions which is thus independent of data sample sizes and support cardinalities of empirical distributions. 
Indeed, suppose
\[
\mu_N := \frac1N \sum\limits_{i=1}^N \delta_{x_i} \quad \textrm{and} \quad \nu_M := \frac1M \sum\limits_{j=1}^M \delta_{y_j}
\]
where $N$ and $M$ are not necessarily equal. 
Then the associated  quantile functions $F_N^{-1}$ and $G_M^{-1}$
are still well-defined as step functions on $(0,1)$. 
Therefore, the quantile correlation coefficient $\rho_q$ and the centered Wasserstein distance $W_2^2(\widetilde{\mu}_N,\widetilde{\nu}_M)$
remain well-defined. 
This also illustrates that the Wasserstein-Taylor diagram can compare probability distributions with  different support cardinalities. In contrast, the classical Taylor diagram requires paired observations and therefore requires the same sample size. 
For example, Figure \ref{fig:WTaylorDifferentSize} shows an example of a Wasserstein-Taylor diagram when sample sizes are different. 
The reference data sequence consists of 10,000 samples from the log normal distribution with parameters $\mu =0$ and $\sigma =1$. The remaining data are: 
\begin{itemize}
    \item[$(1)$] 15,000 samples from the gamma distribution with $2$ and $1$ as shape and scale parameters; 
    \item[$(2)$] 20,000 samples from the log normal distribution with parameters $\mu=0$ and $\sigma =1$;  
    \item [$(3)$]  25,000 samples from the Weibull distribution with $2$ and $1$ as shape and scale parameters; 
    \item [$(4)$] 30,000 samples from the normal distribution with $\mu=0$ and $\sigma =3$;
    \item [$(5)$] 35,000 samples from the exponential distribution with $1$ as scale parameter;
    \item [$(6)$] 40,000 samples from the student $t$ distribution with $4$ degrees of freedom;
    \item [$(7)$] 25,000 samples from the bimodal mixture of two normal distributions with 10,000 samples from $N(2.5, 0.35)$ and 15,000 samples from $N(0.5, 0.15)$.
\end{itemize}

To accommodate data sequences with different sample sizes, we evaluate their empirical quantile functions on a refined common probability grid. The centered $2$-Wasserstein distances and quantile correlation coefficients, which satisfy the law-of-cosines relation involving the standard deviations as stated in Lemma~\ref{WTaylorDiagram}, are then numerically approximated using the corresponding quantile values on this grid.

\begin{figure}[ht]
    \centering
    \includegraphics[scale=0.47]{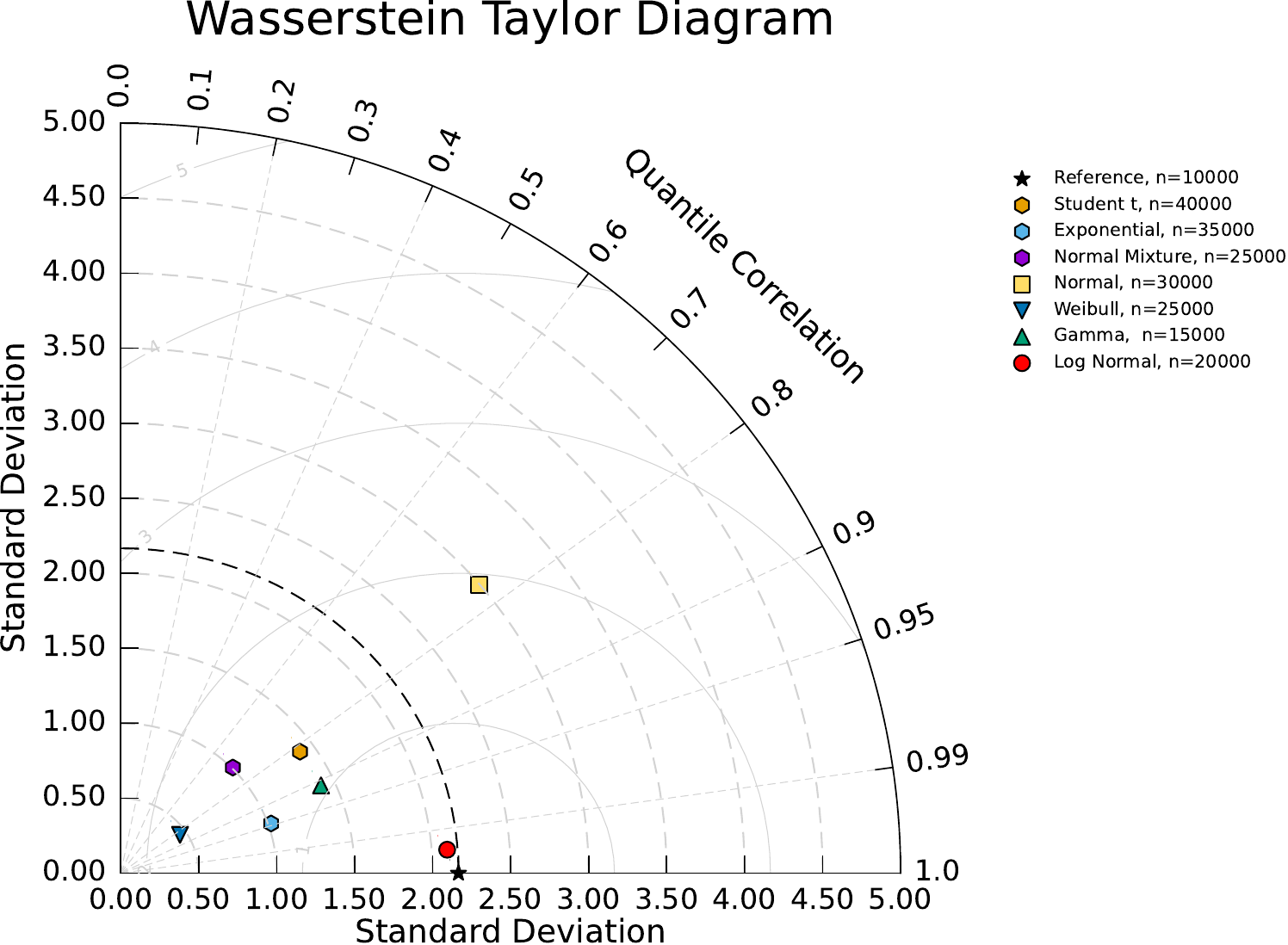} 
    \caption{A Wasserstein-Taylor diagram for samples with different sizes. The reference consists of $10,000$ samples from the log normal distribution with parameters $\mu =0$ and $\sigma =1$. The standard deviation (STD) of the reference is $2.16$. As expected, the data sequence drawn from the same log normal distribution with 20,000 samples performs the best according to the Wasserstein-Taylor-diagram criterion: greatest quantile correlation ($0.997$), smallest centered $2$-Wasserstein distance ($0.173$), and closest STD ($2.099$) to the reference.
    }
    \label{fig:WTaylorDifferentSize}
\end{figure}

There are two distinct correlation coefficients in this paper: the \emph{product-moment correlation coefficient $\rho$} and the \emph{quantile correlation coefficient $\rho_q$}. 
It is well-known that the product-moment correlation coefficient measures the temporal or spatial agreement and is highly sensitive to timing or location errors, while the quantile correlation coefficient measures the similarity of centered distributions and is independent of the temporal or spatial alignment. 
Recalling their definitions in Equations ~\eqref{eqn:defncorrcoeff1} and~\eqref{eqn:defncorrcoeff2}, we have
\begin{align*}
\rho(X,Y) &=  \frac{\langle X', Y' \rangle_{\ell^2(N)}}{ \|X'\|_{\ell^2(N)} \|Y'\|_{\ell^2(N)}} \quad \textrm{for paired data $((x_i,y_i))_{i=1}^N$ with $\sigma_X,\sigma_Y>0$ and}\\
&\hspace{110pt} \textrm{corresponding centered data $X'$ and $Y'$ and}\\
\rho_q(\mu,\nu) &= \frac{\langle \widetilde{F}^{-1}, \widetilde{G}^{-1}\rangle_{L^2(0,1)}}{\| \widetilde{F}^{-1}\|_{L^2(0,1)} \| \widetilde{G}^{-1}\|_{L^2(0,1)}} \quad \textrm{for $\mu, \nu \in \mathcal{P}_2(\mathbb{R})$ with $\sigma_\mu,\sigma_\nu>0$ and}\\
&\hspace{137pt} \textrm{quantile functions $F^{-1}$ and $G^{-1}$.}
\end{align*}

Our goal in the remainder of this section is to relate the quantile correlation applied to finitely supported measures to product-moment correlations of the support of the measures and to quantile correlations of underlying measures when the finitely supported measures are empirical measures.
The result in Equation ~\eqref{eqn:disccorr1} follows from applying the Fr\'echet-Hoeffding upper bound~\cite{hoeffding1940masstabinvariante,frechet1951correlation} to finitely supported distributions and then computing the correlation (see, e.g.,~\cite{lin2014recent}), but we include a direct proof here for completeness.
\begin{proposition}\label{prop:sort}
Let $X=(x_i)_{i=1}^N$ and $Y=(y_i)_{i=1}^N$ be non-constant paired data. Let $X_\circ=(x_{(i)})_{i=1}^N$ and $Y_\circ=(y_{(i)})_{i=1}^N$ be permutations of $X$ and $Y$, respectively, such that
\[
x_{(1)} \le \cdots \leq x_{(N)} \quad  \textrm{and} \quad  y_{(1)} \le \cdots \leq y_{(N)}.
\]
Further define
\[
\mu_X = \frac{1}{N} \sum_{i=1}^N \delta_{x_i} \quad \textrm{and} \quad \nu_Y = \frac{1}{N} \sum_{i=1}^N \delta_{y_i}
\]
Then
\begin{equation} \label{eqn:disccorr1}
\rho(X,Y) \leq \rho(X_{\circ},Y_{\circ}) = \rho_q(\mu_X,\nu_Y),
\end{equation}
where equality holds if and only if $(x_i-x_j)(y_i-y_j) \geq 0$ for all $i,j$, and
\[
W_2(\tilde{\mu}_X,\tilde{\nu}_Y) = E(X'_\circ ,Y'_\circ),
\]
for $\tilde{\mu}_X$, $\tilde{\nu}_Y$ the centered atomic measures and $X'_\circ$, $Y'_\circ$ the centered sorted data.
\end{proposition}
\begin{proof}
    Recall that given two real sequences $(a_i)_{i=1}^N$, $(b_i)_{i=1}^N$ and the corresponding sorted sequences $(a_{(i)})_{i=1}^N$ and $(b_{(i)})_{i=1}^N$, the rearrangement inequality shows that
$$
\sum_{i=1}^N a_i b_i \leq \sum_{i=1}^N a_{(i)} b_{(i)},
$$
and the equality holds if and only if $(a_i-a_j)(b_i-b_j) \geq 0$ for any $i$ and $j$, that is,  $(a_i)$ and $(b_i)$ are \emph{comonotone}. 
Then by the rearrangement inequality, 
$$
\sum_{i=1}^N (x_i-\overline{x})(y_i-\overline{y}) \leq
\sum_{i=1}^N (x_{(i)}-\overline{x})(y_{(i)}-\overline{y}).
$$
That is, $\rho(X,Y) \leq \rho(X_{\circ},Y_{\circ})$ with equality if and only if $(x_i-x_j)(y_i-y_j) \geq 0$ for any $i$ and $j$. 
Note that
\[
m_{\mu_X} = \overline{x} = \overline{x_{\circ}} \quad \textrm{and} \quad m_{\nu_Y} = \overline{y} = \overline{y_{\circ}}.
\]
Thus, the quantile function $\tilde{F}^{-1}$ of the centered measure $\tilde{\mu}_X$ has the form
\begin{align*}  
\tilde{F}^{-1}(t) &= \inf \left\{x\in\mathbb{R}: \tilde{\mu}_X(-\infty, x]\geq t\right\}\\
&=\left\{\begin{array}{lr}x_{(1)}-\overline{x_{\circ}}; & 0 < t \leq 1/N \\ x_{(2)}-\overline{x_{\circ}}; & 1/N < t \leq 2/N \\ \vdots & \\ x_{(N)}-\overline{x_{\circ}}; & (N-1)/N < t < 1 \end{array} \right.,
\end{align*}
where $x_{(i)}$ and $x_{(i+1)}$ may not be distinct for all $i$.
The quantile function $\tilde{G}^{-1}$ of the centered measure $\tilde{\nu}_Y$ is similarly a piecewise flat function taking the values of the sorted, centered vector over evenly sized intervals.
Define the centered, sorted vectors $X'_\circ$ and $Y'_\circ$ by subtracting $\overline{x_{\circ}}$ from $X_\circ$ and $\overline{y_{\circ}}$ from $Y_\circ$, respectively.
Then,
\begin{align*}
\ip{\tilde{F}^{-1}}{\tilde{G}^{-1}}_{L^2(0,1)} &= \frac{1}{N} \ip{X'_\circ}{Y'_\circ}_{\ell^2(N)},\\
\norm{\tilde{F}^{-1}}_{L^2(0,1)} &= \frac{1}{\sqrt{N}} \norm{X'_\circ}_{\ell^2(N)},\\
\norm{\tilde{G}^{-1}}_{L^2(0,1)} &= \frac{1}{\sqrt{N}} \norm{Y'_\circ}_{\ell^2(N)},\\
W_2^2(\tilde{\mu}_X,\tilde{\nu}_Y) &= \int_0^1 \abs{\tilde{F}^{-1}(t)-\tilde{G}^{-1}(t)}^2 \, dt \\
&= \frac{1}{N} \sum_{i=1}^N (x'_{(i)} - y'_{(i)})^2 = E(X'_\circ,Y'_\circ)^2.
\end{align*}
It now follows that $\rho(X_{\circ},Y_{\circ}) = \rho_q(\mu_X,\nu_Y)$ and $W_2(\tilde{\mu}_X,\tilde{\nu}_Y) =E(X'_\circ,Y'_\circ)$, as desired.
\end{proof}

Proposition~\ref{prop:sort} shows that the Wasserstein-Taylor diagram will recover the original Taylor diagram if samples $(x_i)_{i=1}^N$ and $(y_i)_{i=1}^N$ are comonotone: $(x_i-x_j)(y_i-y_j) \geq 0$ for any $i$ and $j$. In this case, the centered RMSE coincides with the centered $2$-Wasserstein distance, while the product-moment correlation coefficient coincides with the quantile correlation coefficient.
In climate science applications, the cumulative climate variables indexed by time are nondecreasing and provide natural examples of comonotone structures, such as cumulative precipitation, accumulated runoff, and cumulative solar radiation. 
Therefore, for the cumulative climate variables indexed by time, the Wasserstein-Taylor diagram just becomes the original Taylor diagram. We summarize the above in the following corollary. 

\begin{corollary}\label{WTaylorAsTayloe}
    The Wasserstein-Taylor diagram recovers the Taylor diagram if the samples $(x_i)_{i=1}^N$ and $(y_i)_{i=1}^N$ are comonotone: for any $i$ and $j$, $(x_i-x_j)(y_i-y_j) \geq 0.$ 
\end{corollary}

Now we show that the quantile correlation of empirical measures converges to the quantile correlation of the measures that were sampled from.
\begin{theorem}\label{thm:samplestat}
   Let $X=(x_i)_{i=1}^N$ and $Y=(y_i)_{i=1}^N$  be i.i.d.\ samples from $\mu, \nu \in \mathcal{P}_2(\mathbb{R})$ with $\sigma_\mu,\sigma_\nu>0$.
   Further set
\[
\mu_N = \frac{1}{N} \sum_{i=1}^N \delta_{x_i} \quad \textrm{and} \quad \nu_N = \frac{1}{N} \sum_{i=1}^N \delta_{y_i}
\]
Assume that $N$ is large enough that $\sigma_{\mu_N}\sigma_{\nu_N} >0$ almost surely.
   Then for
   \[
   \hat{\rho}_{q,N} := \rho(X_{\circ}, Y_{\circ}) = \rho_q (\mu_N, \nu_N), 
   \]
   \[
   \abs{\hat{\rho}_{q,N}-\rho_q(\mu,\nu)} \leq \frac{2 W_2(\mu_N,\mu)}{\sigma_\mu} + \frac{2 W_2(\nu_N,\nu)}{\sigma_\nu}
   \]
   and thus
   \[
   \hat{\rho}_{q,N} \rar \rho_q(\mu,\nu) \enskip \textrm{a.s.\ as $N \rar \infty$}.
   \]
\end{theorem}
\begin{proof}
We note that for any nonzero vectors $a, b$ in a Hilbert space
\[
\norm{\frac{a}{\norm{a}}-\frac{b}{\norm{b}}} \leq \frac{2\norm{a-b}}{\norm{b}}
\]
holds.
Indeed, by applying the triangle inequality twice we have
\begin{align*}  
\norm{\frac{\norm{b}}{\norm{a}}a-b} &\leq \norm{\frac{\norm{b}}{\norm{a}}a-a} + \norm{a-b}= \abs{\frac{\norm{b}}{\norm{a}}-1} \norm{a} + \norm{a-b}\\
&= \abs{\norm{b}-\norm{a}} +  \norm{a-b}\leq 2 \norm{a-b}.
\end{align*}
Let $F^{-1}$ and $\tilde{F}^{-1}$ denote the quantile and centered quantile function, respectively, for $\mu$; $F_N^{-1}$ and $\tilde{F}_N^{-1}$ denote the quantile and centered quantile function, respectively, for $\mu_N$; $G^{-1}$ and $\tilde{G}^{-1}$ denote the quantile and centered quantile function, respectively, for $\nu$; and  $G_N^{-1}$ and $\tilde{G}_N^{-1}$ denote the quantile and centered quantile function, respectively, for $\nu_N$.
Note that each centered quantile function is the orthogonal projection of the corresponding non-centered quantile function onto the orthogonal complement of constant functions in $L^2(0,1)$.
Thus,
\[
\norm{\tilde{F}_N^{-1} - \tilde{F}^{-1}}_{L^2(0,1)} \leq  \norm{F_N^{-1} - F^{-1}}_{L^2(0,1)} = W_2(\mu_N, \mu),
\]
implying that
\[
\norm{\frac{\tilde{F}_N^{-1}}{\sigma_{\mu_N}} - \frac{\tilde{F}^{-1}}{\sigma_{\mu}}}_{L^2(0,1)} \leq \frac{2W_2(\mu_N,\mu)}{\sigma_\mu}
\]
and similarly for $\nu$ and $\nu_N$.
Now, making use of the triangle inequality and the Cauchy-Schwarz inequality, we have
\begin{align*}
    \abs{\hat{\rho}_{q,N} - \rho_q(\mu,\nu)} &= \abs{\ip{\frac{\tilde{F}_N^{-1}}{\sigma_{\mu_N}}}{\frac{\tilde{G}_N^{-1}}{\sigma_{\nu_N}}} - \ip{\frac{\tilde{F}^{-1}}{\sigma_{\mu}}}{\frac{\tilde{G}^{-1}}{\sigma_{\nu}}}}\\
    &= \abs{\ip{\frac{\tilde{F}_N^{-1}}{\sigma_{\mu_N}}}{\frac{\tilde{G}_N^{-1}}{\sigma_{\nu_N}}} - \ip{\frac{\tilde{F}_N^{-1}}{\sigma_{\mu_N}}}{\frac{\tilde{G}^{-1}}{\sigma_{\nu}}} + \ip{\frac{\tilde{F}_N^{-1}}{\sigma_{\mu_N}}}{\frac{\tilde{G}^{-1}}{\sigma_{\nu}}}- \ip{\frac{\tilde{F}^{-1}}{\sigma_{\mu}}}{\frac{\tilde{G}^{-1}}{\sigma_{\nu}}}}\\ 
    &\leq \abs{\ip{\frac{\tilde{F}_N^{-1}}{\sigma_{\mu_N}}}{\frac{\tilde{G}_N^{-1}}{\sigma_{\nu_N}}-\frac{\tilde{G}^{-1}}{\sigma_{\nu}}}} + \abs{\ip{\frac{\tilde{F}_N^{-1}}{\sigma_{\mu_N}}-\frac{\tilde{F}^{-1}}{\sigma_{\mu}}}{\frac{\tilde{G}^{-1}}{\sigma_{\nu}}}}\\
    &\leq \norm{\frac{\tilde{F}_N^{-1}}{\sigma_{\mu_N}}}_{L^2(0,1)}\norm{\frac{\tilde{G}_N^{-1}}{\sigma_{\nu_N}}-\frac{\tilde{G}^{-1}}{\sigma_{\nu}}}_{L^2(0,1)} + \norm{\frac{\tilde{F}_N^{-1}}{\sigma_{\mu_N}}-\frac{\tilde{F}^{-1}}{\sigma_{\mu}}}_{L^2(0,1)}\norm{\frac{\tilde{G}^{-1}}{\sigma_{\nu}}}_{L^2(0,1)}\\
    &=\norm{\frac{\tilde{G}_N^{-1}}{\sigma_{\nu_N}}-\frac{\tilde{G}^{-1}}{\sigma_{\nu}}}_{L^2(0,1)} + \norm{\frac{\tilde{F}_N^{-1}}{\sigma_{\mu_N}}-\frac{\tilde{F}^{-1}}{\sigma_{\mu}}}_{L^2(0,1)}\\
    &\leq \frac{2W_2(\mu_N,\mu)}{\sigma_\mu} + \frac{2W_2(\nu_N,\nu)}{\sigma_\nu},
\end{align*}
as desired.
It is well-known (cf.~\cite{bobkov2019one} and the sources therein) that $W_2(\mu_N,\mu) \rar 0$ a.s.\ for $N \rar \infty$ for all $\mu \in \mathcal{P}_2(\mathbb{R})$, finishing the proof.
\end{proof}
See~\cite{bobkov2019one} for various results concerning the rates of convergence of empirical measures in Wasserstein metric to tighten Theorem~\ref{thm:samplestat} for measures with additional properties.

\section{Applications of Wasserstein-Taylor Diagram}\label{sec:application}

\subsection{Climate Model Evaluation for Convective Cloud Fraction.}\label{sec:climate}

\begin{table*}
\centering
\caption{Performance measures for the Taylor diagram.}
\label{table:TaylorParametersClimate}

\begin{tabular}{lccc}
\toprule
Model & Centered RMSE & Product-Moment Correlation & Prediction STD \\
\midrule
MLP  & 5.086 & 0.7465 & 5.717 \\
EBMregion   & 6.505 & 0.5254  & 3.942 \\
EBM   & 6.745 & 0.4711 & 3.763 \\
LR   & 7.004 & 0.4662 & 1.712 \\
\hline
\end{tabular}
\end{table*}

\begin{table*}
\centering
\caption{Performance measures for the Wasserstein-Taylor diagram.}
\label{table:wTaylorParametersClimate}

\begin{tabular}{lccc}
\toprule
Model & Centered $2$-Wasserstein & Quantile Correlation & Prediction STD \\
\midrule
MLP  & 2.794 & 0.9531 & 5.717 \\
EBMregion   & 4.646 & 0.8694 & 3.942 \\
EBM   & 5.801 & 0.6770 & 3.763 \\
LR   & 6.937 & 0.5017 & 1.712 \\
\hline
\end{tabular}
\end{table*}

\begin{figure}[ht]
    \centering
    \includegraphics[scale=0.47]{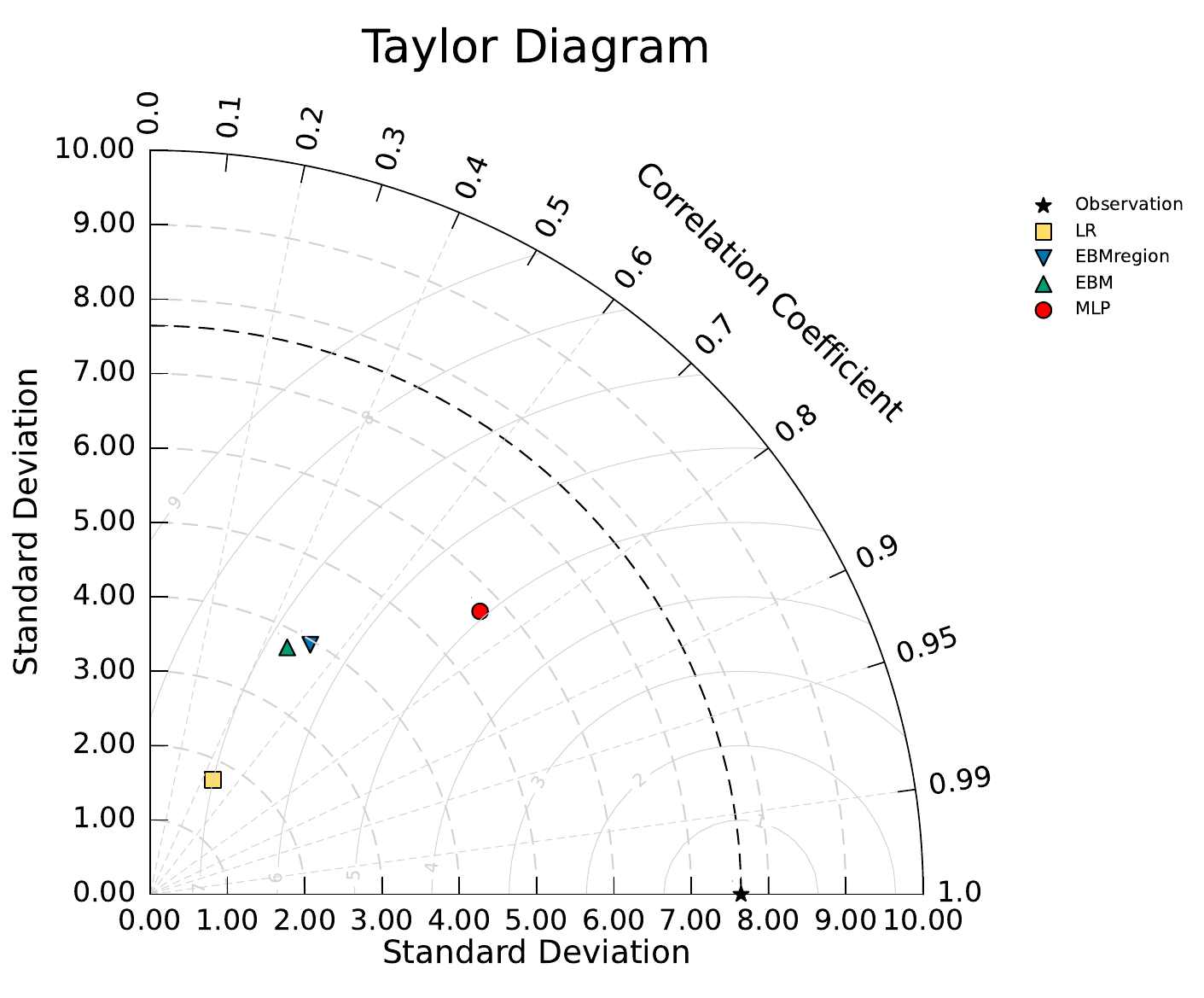} 
    \caption{A Taylor diagram for model evaluation in convective cloud fraction. The standard deviation (STD) of the observed data is 7.645. Note that the MLP model has the closest distance to the reference point.  In fact, it has the smallest centered RMSE, greatest product-moment correlation, and closest prediction STD with the reference data. Therefore, the MLP model has the best performance in temporal alignment compared to other models. } 
\label{fig:TaylorClimate}
\end{figure}

To illustrate a practical application of the Wasserstein-Taylor diagram, we use the hourly convective cloud fraction dataset developed by Yu et al.\ (2026) \cite{yu-2026}. 
The objective is to evaluate the performance of several machine-learning models relative to a ground-truth convective cloud fraction dataset derived from a high-resolution regional climate physical model simulation over the contiguous United States (CONUS), i.e., the CONUS404 dataset \cite{rasmussenetal-2023}. 
The site we choose in this section is located at 25° N and 107° W, and the time range is from 00:00:00 of January 1, 2022 to 23:00:00 of September 30, 2022. 
The models considered include a multilayer perceptron (MLP), an Explainable Boosting Machine (EBM) \cite{lou2013accurate} trained on the full-CONUS-domain dataset, an EBM trained on regional datasets (EBMregion), and a multivariable linear regression model (LR).
EBM is an interpretable machine learning method that allows some insight into why an input produces a particular output.

The performance measures for the associated Taylor diagram and Wasserstein-Taylor diagram are listed in Table \ref{table:TaylorParametersClimate} and  Table \ref{table:wTaylorParametersClimate}, respectively.  
Furthermore, the Taylor diagram and Wasserstein-Taylor diagram can be found in Figure \ref{fig:TaylorClimate} and Figure \ref{fig:wTaylorClimate}. 
The STD of the reference CONUS404 dataset is 7.645. Note that the MLP model has the closest distance to the reference point in the Taylor diagram, and Table \ref{table:TaylorParametersClimate} also shows that the MLP model has the smallest centered RMSE, greatest product-moment correlation, and closest prediction STD with the reference data. 
In this case, the MLP model has the best performance in temporal alignment compared to other models. 
Similarly, the MLP model has the closest distance to the reference point in the Wasserstein-Taylor diagram in Figure \ref{fig:wTaylorClimate}, which means that it has the smallest centered $2$-Wasserstein distance, greatest quantile correlation coefficient, and closest prediction STD with the reference. 
The performance measures in Table \ref{table:wTaylorParametersClimate} also confirm this observation. Therefore, the MLP model has the best performance in the centered probability distribution alignment with reference data compared to other models. 
Combining the results from the Taylor diagram and Wasserstein-Taylor diagram, we conclude that compared to other models, the MLP model has the best performance in the temporal and centered probability distribution alignments with the observations. 

Beyond the MLP results, we note that the two EBM configurations are more clearly separated by the quantile correlation in the Wasserstein-Taylor diagram than by the correlation coefficient in the original Taylor diagram, with EBMregion attaining the higher quantile correlation with the ground truth. 
Since EBMregion is trained exclusively over the target geographic domain, this result suggests that the Wasserstein-Taylor diagram is more sensitive to differences among EBM configurations.
Notably, if one wanted to use an interpretable method (i.e., not an MLP), it is important to clearly distinguish among competing interpretable models. In this use case, the Wasserstein-Taylor diagram can provide such discrimination.

\begin{figure}[ht]
    \centering
    \includegraphics[scale=0.47]{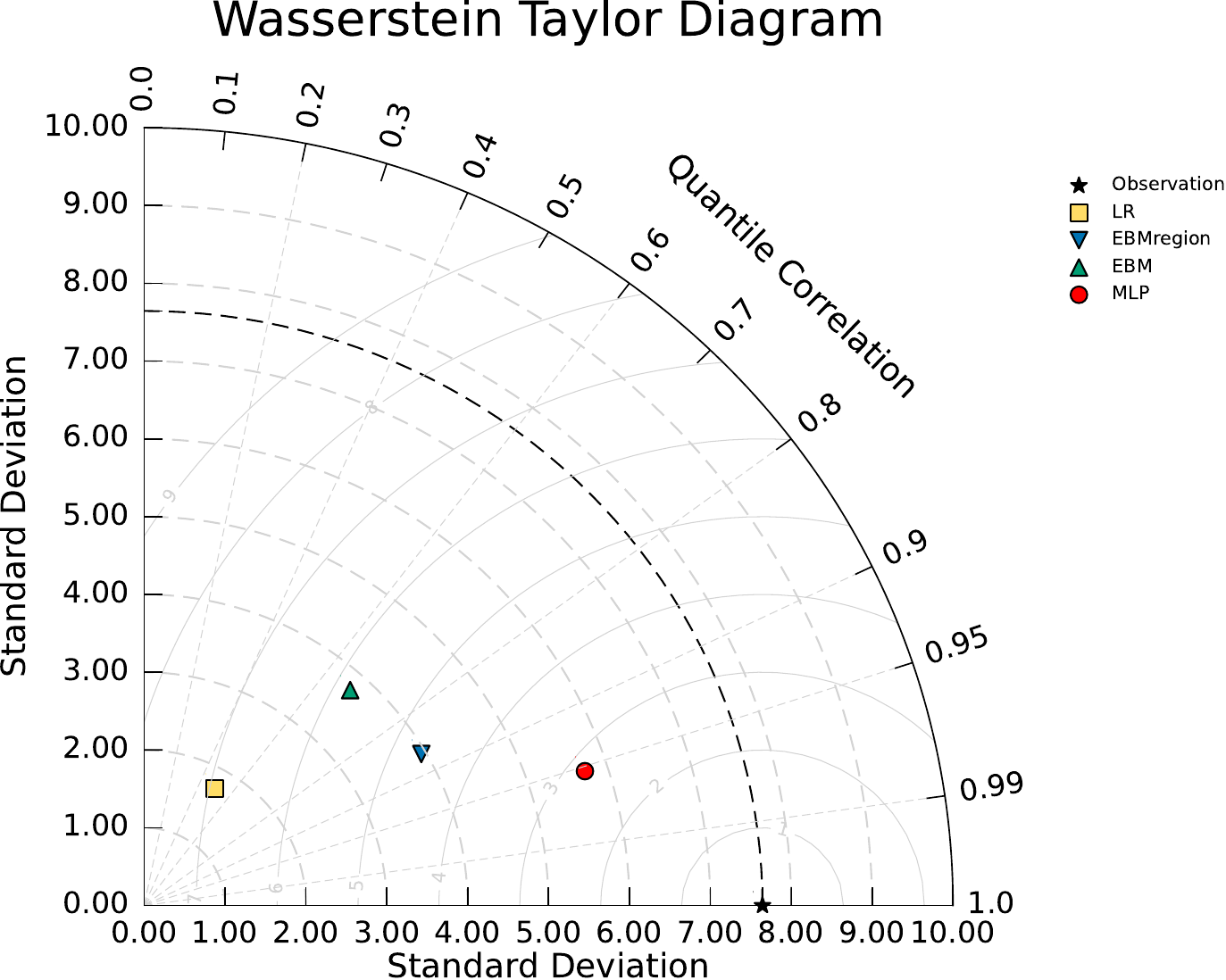} 
    \caption{A Wasserstein-Taylor diagram for model evaluation in convective cloud fraction. The standard deviation (STD) of observed data is 7.645. The MLP model has the closest distance to the reference point, which means that it has the smallest centered $2$-Wasserstein distance, greatest quantile correlation coefficient, and closest prediction STD with the reference. Therefore, the MLP model is the best according to the Wasserstein-Taylor diagram criterion compared to other models.} 
    \label{fig:wTaylorClimate}
\end{figure}

\subsection{Comparison of Probability Distribution Fitting for Wind Speed}\label{sec:wind} A good estimation of the wind speed distribution is of high interest and can contribute to many applications, including assessment of wind energy potential, risk assessment of insurance,  pollutant dispersion control in urban areas, and building and bridge design \cite{yang2025beta}. 
A central question in wind speed distribution estimation is how to compare candidate distributions and choose the best model. 
In this setting, we only have observed wind speed data and potential probability distributions; there are no predictions. 
As explained above, the Wasserstein-Taylor diagram  provides a visualization framework to compare probability distributions in such scenarios. 

\begin{table*}
\centering
\caption{Performance measures and prediction standard deviations (STDs) for considered probability distributions.}
\label{table:wTaylorParameters}

\begin{tabular}{lccc}
\toprule
Model & Centered $2$-Wasserstein & Quantile Correlation & Prediction STD \\
\midrule
BGL  & 0.103 & 0.9996 & 3.488 \\
BE   & 0.217 & 0.9985 & 3.360 \\
BW   & 0.219 & 0.9980 & 3.442 \\
BL   & 0.266 & 0.9980 & 3.314 \\
GL   & 0.610 & 0.9907 & 3.055 \\
L    & 0.447 & 0.9946 & 3.721 \\
LogN & 0.753 & 0.9879 & 3.957 \\
GAM  & 0.518 & 0.9922 & 3.159 \\
W    & 0.624 & 0.9851 & 3.243 \\
GEV  & 1.086 & 0.9662 & 3.972 \\
\bottomrule
\end{tabular}
\end{table*}

\begin{figure}[ht]
    \centering
    \includegraphics[scale=0.47]{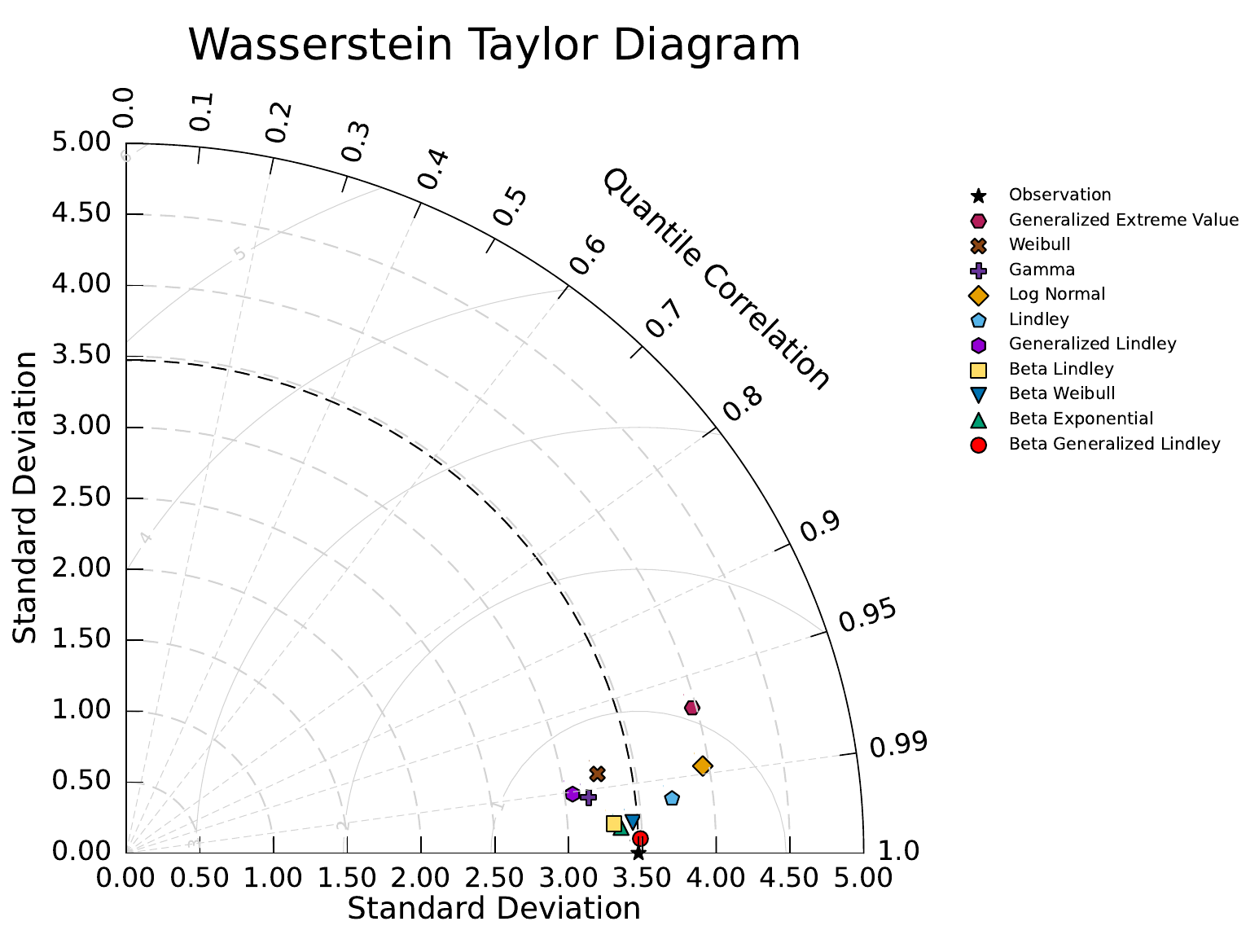} 
    \caption{A Wasserstein-Taylor diagram for distribution comparison of wind speed. The standard deviation of wind speed observation is 3.474, and the beta generalized Lindley distribution has the best performance according to the centered $W_2$ criterion, since it has the smallest distance to the reference point in the Wasserstein-Taylor diagram.} 
    \label{fig:WTaylor_Wind_Speed}
\end{figure}

Following the study in \cite{yang2025beta}, we use distributions to model the wind speed data derived from the measurements at a height of 80 meters of the Flatirons M2 meteorological tower \cite{jager1996nrel}, located about 8 kilometers south of Boulder, Colorado, USA. We use hourly averaged wind speed data from 00:00:00 of January 1, 2010 to 23:00:00 of December 31, 2020. 
Distributions used in this paper include beta generalized Lindley (BGL), beta Lindley (BL), generalized Lindley (GL), Lindley (L), beta Weibull (BW), Weibull (W), beta exponential (BE), log normal (LogN), gamma (GAM), and generalized extreme value (GEV). 
More details about the density functions and parameter estimates can be found in Tables 3-4 in \cite{yang2025beta}.

By calculating the standard deviations and quantiles of each distribution, one can find the associated centered $2$-Wasserstein distances and quantile correlation coefficients with observed wind speed data, as shown in Table \ref{table:wTaylorParameters}. 
The STD of the wind speed observation is 3.474. 
Next, one can generate the corresponding Wasserstein-Taylor diagram as shown in Figure \ref{fig:WTaylor_Wind_Speed}, which shows that the beta generalized Lindley distribution has the smallest distance to the reference point. 
Table \ref{table:wTaylorParameters} also shows that the beta generalized Lindley distribution has the smallest centered $2$-Wasserstein distance, greatest quantile correlation coefficient, and closest STD to the wind speed distribution. 
Therefore, the beta generalized Lindley distribution has the best according to the centered $W_2$ criterion, which also matches the best result stated in \cite{yang2025beta} which used a variety of goodness-of-fit measures.

\subsection{Normalized Wasserstein-Taylor Diagram}\label{sec:normalized}
In some situations, we would like to compare the model performance for different variables or datasets in a single diagram. 
For example, one may want to simultaneously compare the climate model performance for two variables like temperature and precipitation. 
This leads to the so-called \emph{normalized Taylor diagram}. 
Mathematically, dividing Equation \eqref{TaylorDiagram} by $\sigma_x^2>0$ on both sides, we have the following law of cosines
\begin{equation}
\left(\frac{E}{\sigma_x}\right)^2 = 1 + \left(\frac{\sigma_y}{\sigma_x}\right)^2 - 2 \frac{\sigma_y}{\sigma_x} \rho 
 = 1 + \left(\frac{\sigma_y}{\sigma_x}\right)^2 - 2\frac{\sigma_y}{\sigma_x}\cos(\theta). 
\end{equation}
In this case, the standard deviation of the observed data is normalized to one, and the coordinate of reference point in the normalized Taylor diagram is just $(1,0)$. 
Therefore, one can compare the relative model performance across variables or datasets in a single normalized Taylor diagram.

In this section, we introduce the \emph{normalized Wasserstein-Taylor diagram} to compare distributions or samples of different sizes for several variables in a single diagram. For example, one may want to simultaneously compare wind speed distributions at several heights. Mathematically, dividing Equation \eqref{eqn:WTaylorDiagram} by $\sigma_\mu^2>0$ on both sides, we have the following laws of cosines
\begin{equation}\label{eqn:normalizedWTaylor1}
    \left(\frac{W_2(\widetilde{\mu}, \widetilde{\nu})}{\sigma_\mu}\right)^2  = 1+ \left(\frac{\sigma_\nu}{\sigma_\mu}\right)^2 -2 \frac{\sigma_\nu}{\sigma_\mu} \rho_q =1 + \left(\frac{\sigma_\nu}{\sigma_\mu}\right)^2 -2  \frac{\sigma_\nu}{\sigma_\mu} \cos(\theta),
\end{equation}
Similarly, the reference point in the normalized Wasserstein-Taylor diagram is still $(1,0)$. Therefore, one can compare the relative performance of distributions for several variables in a single normalized Wasserstein-Taylor diagram. 

\begin{figure}[ht]
    \centering
    \includegraphics[scale=0.47]{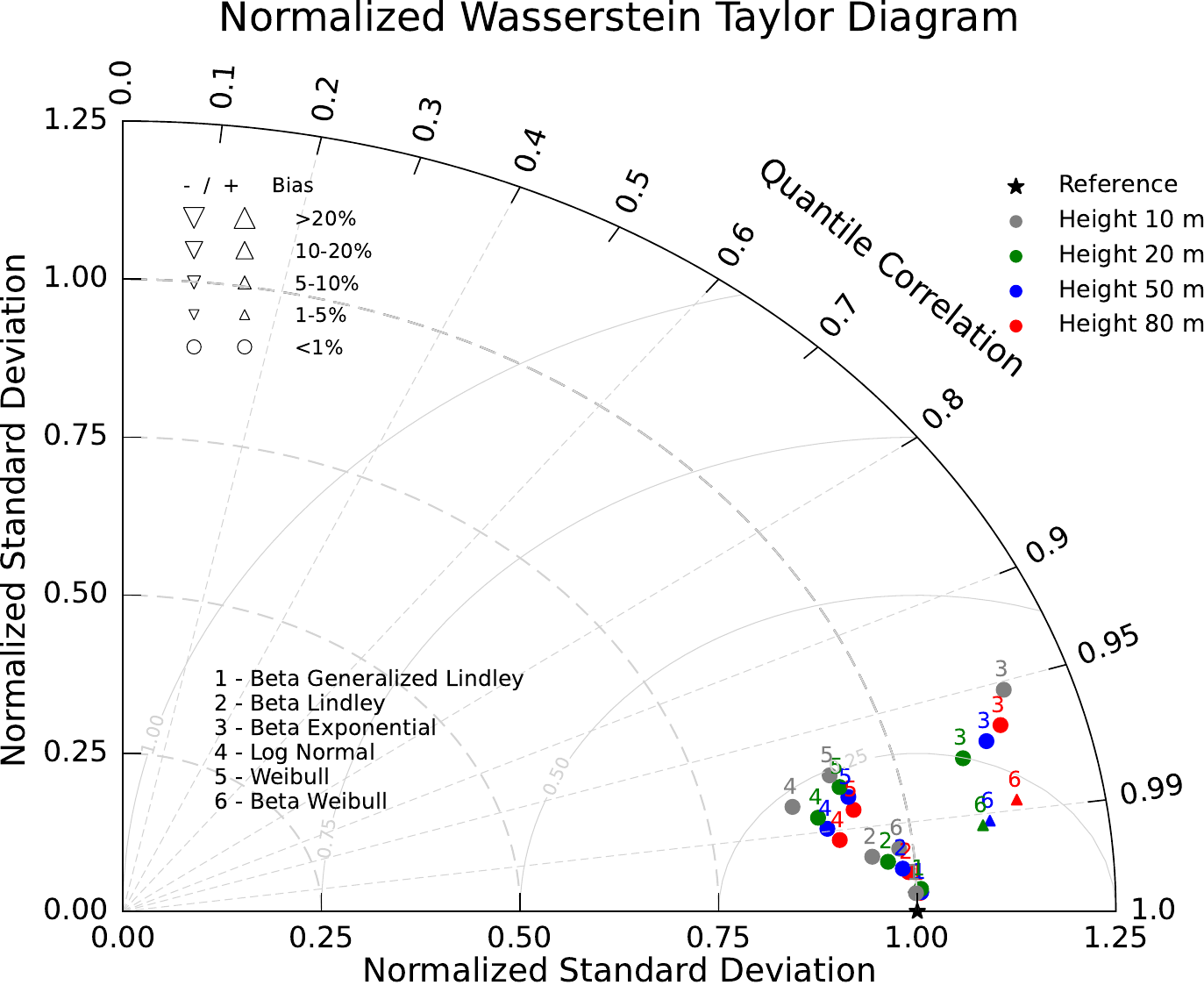} 
    \caption{A normalized Wasserstein-Taylor diagram for wind speed distribution fitting at multiple heights. All the standard deviations of observed data are normalized to one, and the reference point is $(1,0)$.  In the diagram, ``Bias'' means the relative mean bias between the model distribution and the observations: Bias = $(\mu_{\text{model}} - \mu_{\text{obs}})/\mu_{\text{obs}}$. Clearly, among the considered distributions, the beta generalized Lindley distribution has the closest distances to the reference point $(1,0)$ at all heights, which verifies that the beta generalized Lindley distribution has the best performance. } 
    \label{fig:normalizedWTaylor}
\end{figure}

Figure \ref{fig:normalizedWTaylor} shows an example of the normalized Wasserstein-Taylor diagram for wind speed distribution fitting at heights 10, 20, 50, and 80 meters of the Flatirons M2 meteorological tower in the years 2010-2020, and the details about the dataset and distribution parameter estimates can be found in Table 4 of \cite{yang2025beta}. 
As suggested by Equation ~\eqref{eqn:normalizedWTaylor1}, by finding the normalized centered $2$-Wasserstein distance, normalized standard deviations, and quantile correlation coefficients, one can plot the associated diagram satisfying the law of cosines. 
Figure \ref{fig:normalizedWTaylor} shows that among the considered distributions, the beta generalized Lindley distribution has the closest distances to the reference point $(1,0)$ at all heights, which verifies that the beta generalized Lindley distribution is best according to the centered $W_2$ criterion and aligns with the best result in \cite{yang2025beta} which used a variety of goodness-of-fit measures.

\section{Summary and Discussion}\label{sec:summary}
In this work, we modify the framework of the Taylor diagram using the $2$-Wasserstein distance to create the novel Wasserstein-Taylor diagram for model evaluation at the probability distribution level.
 In Lemma \ref{WTaylorDiagram}, we show that the law of cosines still holds in the Taylor diagram if the centered root-mean-squared error (centered RMSE) and product-moment correlation coefficient are replaced by the centered 2-Wasserstein distance  and quantile correlation coefficient, respectively. 
 The Taylor diagram needs observed and predicted data with the same cardinality, so that one can calculate the centered root-mean square error.
However, the Taylor diagram is not applicable when the observed and predicted data sets have different sample sizes, when only observed data are available---as in the case of fitting a statistical distribution to observed wind speed data---or when the observed and predicted data sets are not paired index-by-index.
In these cases one can use the Wasserstein-Taylor diagram since the computation of centered  $2$-Wasserstein distance and quantile correlation coefficient only needs the empirical distribution of observed data and model distributions. 
Therefore, the Wasserstein-Taylor diagram can use the $2$-Wasserstein distance to compare observed and predicted data sets of different sizes from the distribution perspective and even to compare statistical models only with observed data. 
Proposition~\ref{prop:sort} shows that the quantile correlation coefficient of measures assigning $1/N$ times a Dirac delta to points in data sequences of cardinality $N$ is the maximum product-moment correlation coefficient of all of the permutations of those sequences. Then in Corollary \ref{WTaylorAsTayloe}, we show that the Wasserstein-Taylor diagram will recover the original Taylor diagram if given samples are comonotone. Finally, we show in Theorem~\ref{thm:samplestat} that the quantile correlation coefficients of empirical distributions converge to the quantile correlation of the underlying ground truth distributions.

In summary, the Taylor diagram compares the spatial or temporal alignment of data while the Wasserstein-Taylor diagram compares the similarity from the probability distribution perspective. 
Both methods rely on the data / measures being centered first.
Furthermore, as proved in Proposition~\ref{prop:sort} and claimed in  Corollary \ref{WTaylorAsTayloe},  the Wasserstein-Taylor diagram will recover the original Taylor diagram if samples $(x_i)_{i=1}^N$ and $(y_i)_{i=1}^N$  from $\mu, \nu \in \mathcal{P}_2(\mathbb{R})$  are comonotone: $(x_i-x_j)(y_i-y_j) \geq 0$ for any $i$ and $j$, which is often seen in the cumulative climate variables indexed by time such as cumulative precipitation, runoff, and solar radiation. 
We also introduce the so-called normalized Wasserstein-Taylor diagram to compare the relative performance of distributions for
several variables in a single diagram. 

Recall the decomposition of $W_2$ in Equation ~\eqref{eq:W2Decomposition} from~\cite{irpino2007optimal}.
In Lemma~\ref{WTaylorDiagram}, we required the compared distributions to be centered in order to get a law of cosines, which implicitly meant that we only focused on comparing the size and shape of the distributions.
A natural modification of Wasserstein-Taylor diagrams would be to indicate information comparing locations of the non-centered distributions  (i.e., $m_\mu - m_\nu$) using different sizes or face colors for the markers.

\section*{Data Availability}
The code and data used in this paper can be found at \url{https://github.com/DowellChan/Wasserstein_Taylor_Diagram}.

\section*{Acknowledgment}
DC would like to thank Zhiang Xie, Puxi Li, and Senfeng Liu for helpful discussions.  
EJK and HY were supported by the U.S. National Science Foundation (NSF) Collaborations 
in Artificial Intelligence and Geosciences (CAIG) program (Grant number: 2425923). 
EJK used ChatGPT to modify Figure~\ref{fig:Taylor-example}, generate Figure~\ref{fig:lawcos}, look up literature, and provide feedback on the penultimate draft. 
While using ChatGPT to search for literature on results concerning quantile correlation, EJK discovered the key ideas of Theorem~\ref{thm:samplestat} and incorporated them into the paper.
All mathematical claims, citations, and computations were checked by the authors, who take full responsibility for the content of this paper.

\section*{Competing Interests}
The authors declare no competing interests.

\bibliographystyle{plain} 
\bibliography{refs}
\end{document}